\documentclass[a4paper,11pt]{article}
\usepackage{amsmath,amssymb,amsfonts,amsthm}
\usepackage[english]{babel}
\usepackage{mathtools}
\usepackage{url}

\usepackage[affil-it]{authblk}

\newtheorem{prop}{Proposition}[section]

\title{Ge\-o\-me\-try a\-nd phy\-si\-cs
of pse\-udo\-dif\-fe\-ren\-ti\-al o\-pe\-ra\-to\-rs 
on ma\-ni\-fol\-ds}
\author[1]{Giampiero Esposito}
\author[2]{George M. Napolitano}
\affil[1]{INFN, Sezione di Napoli, Complesso Universitario di Monte S. Angelo,
Via Cintia Edificio 6, 80126 Napoli, Italy}

\affil[2]{Centre for Mathematical Sciences, Lund University, SE-22100 Lund, Sweden}

\begin{document}

\maketitle

\begin{abstract}
A review is made of the basic tools used in mathematics to define a calculus for
pseudodifferential operators on Riemannian manifolds endowed with a connection: 
existence theorem for the function that generalizes the phase; analogue of Taylor's 
theorem; torsion and curvature terms in the symbolic calculus; the two kinds of derivative
acting on smooth sections of the cotangent bundle of the Riemannian manifold; the concept
of symbol as an equivalence class. Physical motivations and applications 
are then outlined, with emphasis on Green functions of quantum field theory and
Parker's evaluation of Hawking radiation.
\end{abstract}

\section{Introduction}
\label{sec:intro}

In the course of studying partial differential equations on $\mathbb{R}^{n}$, one discovers one can
consider operators whose action is defined, at least formally, by the integral 
\begin{equation}
(Pu)(x) \equiv (2 \pi)^{-{\frac{n}{2}}} \int_{\mathbb{R}^{n}}
{\rm e}^{{\rm i}\xi \cdot x}p(x,\xi){\hat u}(\xi){\rm d}\xi,
\label{(1)}
\end{equation}
where ${\rm d}\xi$ is the Lebesgue measure on $\mathbb{R}^{n}$, $p(x,\xi)$ is the {\it amplitude}
of the operator $P$, $\xi \cdot x=\langle x,\xi \rangle$ is its {\it phase} function, and
\begin{equation}
{\hat u}(\xi)=(2 \pi)^{-{\frac{n}{2}}}\int_{\mathbb{R}^{n}}
{\rm e}^{-{\rm i}\xi \cdot x}u(x){\rm d}x
\label{(2)}
\end{equation}
is the Fourier transform of $u$. These operators are said to be pseudodifferential and form a
class large enough to contain the differential operators, the Green operators and the singular
integral operators that are used to solve partial differential equations. They also contain,
for each elliptic operator, its parametrix, i.e. an approximate inverse up to an operator of
lower order (see below).

To understand how pseudo-differential operators can be used in solving inhomogeneous partial 
differential equations, let us begin by considering, for $n  \geq 3$ and a smooth function 
$f$ with compact support, i.e. $f \in C_0^{\infty}(\mathbb{R}^n)$,   
the inhomogeneous equation $\bigtriangleup u=f$, where
$$
\bigtriangleup \equiv \sum_{k=1}^{n}{\frac{\partial^{2}}{\partial x_{k}^{2}}}
$$
is minus the Laplacian (with our convention, $\bigtriangleup$ is defined in the standard way, but
the Laplacian has symbol given by $\sum_{k=1}^{n}(\xi_{k})^{2}=|\xi|^{2}$). Upon taking the Fourier 
transform of both sides, one finds
\begin{equation}
-|\xi|^{2}{\hat u}(\xi)={\hat f}(\xi),
\label{(3)}
\end{equation}
and hence, for the inverse operator $Q$ of $\bigtriangleup$, or fundamental solution, we 
can write \cite{ref:boo,ref:grubb}
\begin{equation}
\begin{split}
u(x)=(Qf)(x) & =-(2\pi)^{-\frac{n}{2}}\int_{\mathbb{R}^n}
{\rm e}^{{\rm i}\xi \cdot x} \frac{{\hat f}(\xi)}{|\xi|^{2}} {\rm d}\xi\\
& = - \frac{\Gamma\left(\frac{n}{2}-1\right)}{4\pi^{\frac{n}{2}}} \int_{\mathbb{R}^n} \frac{f(y)}{|x-y|^{n-2}} {\rm d}y
\end{split}
\label{(4)}
\end{equation}

If, instead of the Laplacian on $\mathbb{R}^n$, we deal with a general partial differential operator
with constant coefficients which can be written as a polynomial $P=p(D)$, where
$D \equiv -{\rm i}\left({\frac{\partial}{\partial x_{1}}},...,{\frac{\partial}{\partial x_{n}}} \right)$, so that
\begin{equation}
p(D)u=f \in C_{0}^{\infty}(\mathbb{R}^n),
\label{(5)}
\end{equation}
the solution can be formally expressed as
\begin{equation}
u(x)=(Qf)(x)=(2\pi)^{-{\frac{n}{2}}}\int_{\mathbb{R}^n}
{\rm e}^{{\rm i}\xi \cdot x}q(\xi){\hat f}(\xi){\rm d}\xi,
\label{(6)}
\end{equation}
where the amplitude $q(\xi)$ is the inverse of the symbol $p(\xi)$ of $p(D)$, and the integration
contour must avoid the zeros of $p(\xi)$. However, if the operator $P$ has variable coefficients
on a subset $U$ of $\mathbb{R}^n$, i.e.
\begin{equation}
P=p(x,D)=\sum_{\alpha: |\alpha| \leq k} a_{\alpha}(x)D_{x}^{\alpha}, \;
a_{\alpha} \in C^{\infty}(U),
\label{(7)}
\end{equation}
one can no longer solve the equation $Pu=f$ by Fourier transform. One can however {\it freeze}
the coefficients at a point $x_{0} \in U$ and consider $P$ as a perturbation of $p(x_{0},D)$, which is
hence a differential operator with constant coefficients. In this way the amplitude of $P$ reduces
to $q(\xi)=\frac{1}{p_{0}(x_{0},\xi)} $, and if we let $x_{0}$ vary in $U$, we obtain the approximate
solution operator
\begin{equation}
(Qf)(x)=(2\pi)^{-{\frac{n}{2}}}\int_{\mathbb{R}^n}
{\rm e}^{{\rm i}\xi \cdot x}q(x,\xi){\hat f}(\xi){\rm d}\xi
\label{(8)}
\end{equation}
with amplitude $q(x,\xi)=p(x,\xi)^{-1}$, for $x \in U, f \in C_{0}^{\infty}(\mathbb{R}^n)$.

However, to solve the inhomogeneous equation $Pu=f$, we do not strictly need the full inverse operator
or fundamental solution, but, as we said before, it is enough to know a parametrix, i.e. a quasi-inverse
modulo a regularizing operator. One can provide a first example of parametrix by reverting to the study
of constant coefficient operators. 
A parametrix $Q$ can then be constructed by choosing its amplitude \cite{ref:boo} as
\begin{equation}
q(\xi)=\frac{\chi(\xi)}{p(\xi)},
\label{(9)}
\end{equation}
where $\chi(\xi)$ is a suitably chosen $C^{\infty}(\mathbb{R}^n)$ function which is identically zero in a disk about the
origin and identically $1$ for large $\xi$. In this way the integral formula
\begin{equation}
(Qf)(x)=(2\pi)^{-{\frac{n}{2}}}\int_{\mathbb{R}^n} {\rm e}^{{\rm i}\xi \cdot x} q(\xi){\hat f}(\xi){\rm d}\xi,
\label{(10)}
\end{equation}
is not affected by the convergence problems that would be met if the amplitude were taken
to be just $p(\xi)^{-1}$. If $Q$ is the integral operator in (\ref{(10)}) it is no longer true that
$PQ=QP=I$, but we have \cite{ref:boo}
\begin{equation}
PQf=f+Rf, \; f \in C^{\infty}(U),
\label{(11)}
\end{equation}
\begin{equation}
(Rf)(x) \equiv (2\pi)^{-{\frac{n}{2}}}\int_{\mathbb{R}^n}r(x-\xi)f(\xi){\rm d}\xi,
\label{(12)}
\end{equation}
where the Fourier transform of $r$ is $\chi-1$. Thus $r$ is a smooth function and $R$ turns out to be
a smoothing operator. Such a class of smoothing operators is fully under control, and hence a
parametrix $Q$ serves just as well as a full fundamental solution, for which $R$ vanishes identically.

In the following we will see how to extend the theory of pseudodifferential operators on $\mathbb{R}^n$ to 
more general pseudodifferential operators defined on compact manifolds. In particular, we will outline some 
basic symbolic calculus for such operators. In the last two sections of this paper, some applications to physics will be shown. 

\section{Pseudodifferential operators on manifolds} 

First, note that the material in the appendix can be re-expressed by saying that a pseudodifferential 
operator $A$ acting on functions in $C_0^\infty(\mathbb{R}^n)$ has a
symbol given by
\begin{equation}
\sigma_{A}(x,\xi)={\rm e}^{-{\rm i}\xi \cdot x} A {\rm e}^{{\rm i}\xi \cdot x}
\iff \sigma_{A}(x_{0},\xi)
= \left. \Bigr[A {\rm e}^{{\rm i}\xi \cdot (x-x_{0})}\Bigr] \right|_{x=x_{0}},
\label{(13)}
\end{equation}
and hence formal application of Taylor's formula yields
\begin{equation}
{\rm e}^{-{\rm i}\xi \cdot x}A f(x) {\rm e}^{{\rm i}\xi \cdot x}
=\sum_{k_{1} \dots k_{n}} \frac{{\rm i}^{-(k_{1}+ \cdots +k_{n})}}{k_{1}! \dots k_{n}!}
\left({\partial^{k_{1}}\over \partial \xi_{1}^{k_{1}}} \cdots
 {\partial^{k_{n}}\sigma_{A}\over \partial \xi_{n}^{k_{n}}}
 \right)
 \left({\partial^{k_{1}}\over \partial x_{1}^{k_{1}}} \cdots
 {\partial^{k_{n}}f\over \partial x_{n}^{k_{n}}} \right),
\label{(14)}
\end{equation}
the sum being taken over all values of the multi-index $k=(k_{1},...,k_{n})$. The symbol of the
product $AB$ of two pseudodifferential operators $A$ and $B$ is then defined by
\begin{equation}
\begin{split}
\sigma_{AB}(x,\xi)& = {\rm e}^{-{\rm i}\xi \cdot x} AB {\rm e}^{{\rm i}\xi \cdot x}
={\rm e}^{-{\rm i}\xi \cdot x}A {\rm e}^{{\rm i}\xi \cdot x}
{\rm e}^{-{\rm i}\xi \cdot x}B {\rm e}^{{\rm i}\xi \cdot x} \\
&= {\rm e}^{-{\rm i}\xi \cdot x}A {\rm e}^{{\rm i}\xi \cdot x}\sigma_{B} \\
&= \sum_{k_{1} \dots k_{n}} \frac{{\rm i}^{-(k_{1}+ \cdots +k_{n})}}{k_{1}! \dots k_{n}!}
\left({\partial^{k_{1}}\over \partial \xi_{1}^{k_{1}}} \cdots
 {\partial^{k_{n}}\sigma_{A}\over \partial \xi_{n}^{k_{n}}}
 \right)
 \left({\partial^{k_{1}}\over \partial x_{1}^{k_{1}}} \cdots
 {\partial^{k_{n}}\sigma_{B}\over \partial x_{n}^{k_{n}}} \right)
\end{split}
\label{(15)}
\end{equation}
	
Note here the crucial role played by the function
\begin{equation}
l(x_{0},\xi,x)=\xi \cdot (x-x_{0})=\sum_{l=1}^{n}\xi_{l}(x^{l}-x_{0}^{l})
\label{(16)}
\end{equation}
which, for each $x_{0}$, is linear in $x$ and $\xi$. Its derivative with respect to $x$ is $\xi$,
while its derivative with respect to $\xi$ is $x-x_{0}$.

On going from $\mathbb{R}^n$ to compact manifolds, we look for a real-valued function 
\begin{align*}
	l \, :  \, T^*M \times M & \to \mathbb{R} \\
	(v,x) & \mapsto l(v,x),
\end{align*}
which generalizes the function (\ref{(16)}). Linearity in $\xi$ becomes linearity in $v$ on each fiber 
of the cotangent bundle of $M$, but linearity in
$x$ has no obvious counterpart. However, if there exists a connection $\nabla$ on $T^{*}M$, linearity
at $x_{0}$ can be defined by stating that, for all integer $k \geq 2$, the symmetrized $k$-th covariant 
derivative vanishes at $x_{0}$. The desired linear function is then a real-valued function 
$l \in C^{\infty}(T^{*}M \times M)$ such that the image $l(v,x)$ is, for fixed $x$, linear in each
fiber of $T^{*}M$, and such that, for each $v \in T^{*}M$,
\begin{equation}
\partial^{k}l(v,x)|_{x=\pi(v)}=
\begin{cases}
	v \quad \text{ if } k =1 \\
	0 \quad \text{ otherwise}.
\end{cases}
\label{(17)}
\end{equation}
With this notation, $\partial^{k}$ is the symmetrized $k$-th covariant derivative with respect to $x$,
while $\pi$ is the projection map $\pi: T^{*}M \rightarrow M$. 

Once a connection $\nabla$ is assigned, a definition of symbol of a pseudodifferential operator $A$
on $C^{\infty}(M)$ is provided by \cite{ref:widom}
\begin{equation}
\sigma_{A}(v)= \left. \Bigr[A {\rm e}^{{\rm i}l(v,x)}\Bigr] \right|_{x=\pi(v)},
\label{(18)}
\end{equation}
which is a generalization of formula (\ref{(13)}). However, the function $l(v,x)$, whose existence will be proved 
in the next section, is not uniquely determined by the linearity conditions above. Therefore, different functions 
$l$ would lead to different symbol maps. On the other hand, it can be proved \cite{ref:widom} that the difference 
between any two symbol maps corresponding to different choices of functions $l$ belongs to a certain class of functions, 
therefore the symbol of a pseudodifferential operator will be actually defined as an element 
of a quotient space suitably defined.

\section{Existence theorem for the function $l$}
\label{sec:l}

Following our main source \cite{ref:widom}, we are now going to prove that there exists a function
$l \in C^{\infty}(T^{*}M \times M)$ with the properties listed above. Indeed, we have the following result. 

\begin{prop}[\cite{ref:widom}]
There exists a function $l \in C^{\infty}(T^{*}M \times M)$ such that $l(\cdot,x)$ is, 
for each $x \in M$, linear on the fibers
of $T^{*}M$ and such that, for each cotangent vector $v \in T^{*}M$, eq. (\ref{(17)}) holds.
\end{prop}

\begin{proof}
The desired function $l$ is first constructed locally. If $U$ is a coordinate neighbourhood in $M$
with local coordinates $x^{i}$ and an $m$-th order covariant tensor
$\tau_{i_{1} \dots i_{m}}$, its covariant derivative is given by
\begin{equation}
\tau_{i_{1} \dots i_{m};i} = \frac{\partial}{\partial x^{i}}  \tau_{i_{1} \dots i_{m}}
-\sum_{\nu}\Gamma_{\; i i_{\nu}}^{j}
\; \tau_{i_{1} \dots i_{\nu -1}j i_{\nu+1} \dots i_{m}},
\label{(19)}
\end{equation}
where $\Gamma_{\; ik}^{j}$ is the standard notation for Christoffel symbols. Thus by induction, for any
scalar function $f$, one has 
\begin{equation}
f_{;i_{1} \dots ;i_{k}}= \frac{\partial^{k}f}{\partial x^{i_{1}} \cdots  \partial x^{i_{k}}} 
+\sum_{j: |j| <k}\gamma_{i_{1} \dots i_{k}j} \frac{\partial^{j}f}{\partial x^{j}} ,
\label{(20)}
\end{equation}
where $j$ are multiindices and $\gamma$'s are polynomials in the derivatives of Christoffel symbols.
This implies that, for $k>1$, condition (\ref{(17)}) is equivalent to each term
$$
\frac{\partial^{k}l}{\partial x^{i_{1}} \dots \partial x^{i_{k}}}
$$
being equal, at $\pi(v)$, to some linear combination of lower-order derivatives. To sum up, starting
with the requirements
\begin{equation}
l(v,\pi(v))=0, \qquad 
\left .  \frac{\partial l(v,x)}{\partial x^{i}} \right |_{x=\pi(v)}
=v \left( \frac{\partial}{\partial x^{i}}  \right),
\label{(21)}
\end{equation}
specifies what $\left . \frac{\partial^{k}l}{\partial x^{k}} \right|_{x = \pi(v)}$ must be, in order eq.
(\ref{(17)}) to hold. Borel's theorem ensures that there exists a $C^{\infty}$ function having partial derivatives 
arbitrarily prescribed, and the proof shows that $l$ may be chosen to be both linear and $C^{\infty}$ 
in $v$ \cite{ref:widom}.

Having established that, for each coordinate neighbourhood $U_{i}$ in $M$, there exists an
$l_{i} \in C^{\infty}(T^{*}U_{i} \times U_{i})$ with the desired properties, we can take finitely many
$U_{i}$ covering $M$ with a partition of unity given by functions 
$\varphi_{i} \in C_{0}^{\infty}(U_{i})$, and yet other smooth functions with compact support 
$\psi_{i} \in C_{0}^{\infty}(U_{i})$ equal to $1$ on a neighbourhood of the support of $\varphi_{i}$.
The function
\begin{equation}
l(v,x) \equiv \sum_{i}\varphi_{i}(\pi(v))l_{i}(v,x)\psi_{i}(x)
\label{(22)}
\end{equation}
is then globally defined and satisfies all requirements.
\end{proof}

\section{Analogue of Taylor's theorem}

For a smooth function $f$ on $M$, $f^{(k)}(x_{0})$ is replaced by $\nabla^{k}f(x_{0})$. To obtain the 
analogue of $x-x_{0}$ note that, since for fixed $x_{0}$ and $x \in M$ the function $l(v,x)$ is linear 
for $v \in T_{x_{0}}^{*}$, we may think of it as an element of the tangent space at $x_{0}$. We can
instead regard $l(\cdot,x)$ as a vector field on $M$, so that
\begin{equation}
l(\cdot,x)^{k}=l(\cdot,x)\otimes  \dots  \otimes l(\cdot,x)
\label{(23)}
\end{equation}
is a symmetric $k$-th order contravariant tensor field, and hence
$$
\nabla^{k}f(x_{0}) \cdot l(x_{0},x)^{k}
$$
is defined and, by virtue of symmetry of $l(x_{0},x)^{k}$, it coincides with
$\partial^{k}f(x_{0}) \cdot l(x_{0},x)^{k}$. A basic theorem \cite{ref:widom}
holds according to which, for each point
$x_{0} \in M$ and each integer $N$, one has
\begin{equation}
\partial^k f(x_0) = \left. \partial^k \sum_{n=0}^{N}\frac{1}{n!} \nabla^{n} f(x_{0}) \cdot l(x_{0},x)^{n} \right|_{x=x_0},
\label{(24)}
\end{equation}
for $k \leq N$.
 
\section{Torsion and curvature terms in the symbolic calculus}

In the symbolic calculus, one encounters frequently the unsymmetrized covariant derivatives
\begin{equation}
\nabla^{k}l(v) \equiv \left. \nabla^{k}l(v,x) \right|_{x=\pi(v)},
\label{(25)}
\end{equation}
which turn out to be polynomials in the torsion tensor $T_{\; ij}^{p}$ and curvature tensor
$R_{\; ijk}^{p}$. Indeed, by virtue of the Ricci identity, the difference of second covariant 
derivatives of an $m$-th order covariant tensor $\tau_{i_{1} \dots i_{m}}$ is given by
\begin{equation}
\tau_{i_{1} \dots i_{m};j;k}-\tau_{i_{1} \dots i_{m};k;j}
=\sum_{\nu}\tau_{i_{1} \dots i_{\nu-1}p i_{\nu+1} \dots i_{m}}
\; R_{\; i_\nu jk}^{p} -\tau_{i_{1} \dots i_{m};p}T_{\; jk}^{p}.
\label{(26)}
\end{equation}
Thus, for any permutation $\alpha$ of $1,\dots,k$, the difference
$$
f_{;i_{1} \dots ;i_{k}}-f_{;i_{{\alpha}(1)};\dots;i_{{\alpha}(k)}}
$$
is a sum of terms each of which is a product of $(k-1)$st or lower-order covariant derivatives of $f$
and covariant derivatives of $T$ and $R$ followed by contraction. In particular, if $f=l$, one has at
$x=\pi(v)$
\begin{equation}
\sum_{\alpha}l_{;i_{\alpha(1)}\dots;i_{\alpha(k)}}=0.
\label{(27)}
\end{equation}
In particular, one finds, using the Einstein summation convention,
\begin{equation}
(\nabla^{2}l)_{ij}=l_{;ij}={1 \over 2}v_{p}T_{\; ij}^{p},
\label{(28)}
\end{equation}
and, in the case of vanishing torsion,
\begin{equation}
(\nabla^{3}l)_{ijk}=l_{;ijk}={1 \over 3}v_{p}\Bigr(R_{\; ijk}^{p}+R_{\; jik}^{p}\Bigr).
\label{(29)}
\end{equation}

\section{The derivatives $D^{k}$ and $\nabla^{k}$}

Given a function $\sigma \in C^{\infty}(T^{*}M)$, one defines $D^{k}\sigma$ to be the $k$-th derivative
of $\sigma$ in the direction of fibers of $T^{*}M$. Thus, for $\pi(v)=x_{0}$, think of $\sigma$ as a
function on the cotangent space at $x_{0}$, and take its $k$-th derivative $D^{k}\sigma$ evaluated at $v$.
This is a $k$-linear function on $T^{*}x_{0}$ and may be identified with an element of the tensor product
$\otimes_{k}T_{x_{0}}$. This means that $D^{k}\sigma$ is a contravariant $k$-tensor, and it is the analogue
of ${\partial^{k}\sigma \over \partial \xi^{k}}$, but of course $k$ is an integer in $D^{k}$ and a multi-index
in ${\partial^{k}\over \partial \xi^{k}}$.

The covariant derivatives $\nabla^{k}$ act on $C^{\infty}(M)$, and to define their action on 
$\sigma \in C^{\infty}(T^{*}M)$ we set \cite{ref:widom}
\begin{equation}
\nabla^{k}\sigma(v) \equiv \left . \nabla^{k}\sigma \Bigr({\rm d}_{x}l(v,x)\Bigr)\right|_{x=\pi(v)}.
\label{(30)}
\end{equation}
Although the function $l(v,\cdot)$ is not unique as we said before, all its derivatives are determined at
$\pi(v)$, and hence (\ref{(30)}) defines $\nabla^{k}\sigma$ unambiguously as a covariant $k$-tensor. The mixed derivatives
$\nabla^{k}D^{j}\sigma$ may also occur and are defined by
\begin{equation}
\nabla^{k}D^{j}\sigma(v) \equiv \nabla_{x}^{k}D_{v}^{j}
\left . \sigma({\rm d}_{x}l(v,x))\right|_{x=\pi(v)}.
\label{(31)}
\end{equation}
This is a contravariant (resp. covariant) $j$-tensor (resp. $k$-tensor), or tensor of type $(j,k)$. 
For example, given a Riemannian manifold $(M,g)$, if $\sigma$ is the squared norm of $v$, i.e.
\begin{equation}
\sigma(v)=|v|^{2}=g^{ij}v_{i}v_{j},
\label{(32)}
\end{equation}
one has 
\begin{equation}
(D \sigma)^{i}=2g^{ij}v_{j}, \; (D^{2}\sigma)^{ij}=2g^{ij}.
\label{(33)}
\end{equation}
Moreover, since
\begin{equation}
\sigma({\rm d}l)=g^{ij}l_{i}l_{j},
\label{(34)}
\end{equation}
and both $l_{i}$ and $g^{ij}$ have vanishing covariant derivatives at $\pi(v)$, one finds
\begin{equation}
\nabla \sigma=0,
\label{(35)}
\end{equation}
while, by virtue of (\ref{(29)}),
\begin{equation}
(\nabla^{2}\sigma)_{kl}=\frac{2}{3} v_{p}v_{q}
\Bigr(R_{\; \; \; kl}^{pq}
+R_{\; k \; \; l}^{p \; \; q}\Bigr),
\label{(36)}
\end{equation}
and
\begin{equation}
(\nabla^{2}D \sigma)_{\; kl}^{p}={4 \over 3}v_{q}
\Bigr(R_{\; \; \; kl}^{pq}+R_{\; k \; \; l}^{p \; \; q}\Bigr).
\label{(37)}
\end{equation}
Note also that $\sigma(v)=|v|^{2}$ is the symbol of the Laplacian on a Riemannian manifold, if the
connection with respect to which the symbol is taken is the Levi-Civita connection.

\section{The symbol as an equivalence class}

Let $S_{\rho}^{\omega}(\mathbb{R}^n\times \mathbb{R}^{m})$, for $1/2 < \rho \leq 1$, be the space of 
functions $\sigma \in C^{\infty}(\mathbb{R}^n \times \mathbb{R}^{m})$ such that, for all
multi-indices $j$ and $k$, 
\begin{equation}
{\partial^{j}\over \partial x^{j}}{\partial^{k}\over \partial \xi^{k}}
\sigma(x,\xi) = O \left( (1+|\xi|)^{\omega-\rho|k|+(1-\rho)|j|} \right),
\label{(38)}
\end{equation}
uniformly on compact $x$-sets.

On going from $\mathbb{R}^n$ to a (compact) manifold $M$, the spaces
$S_{\rho}^{\omega}(M \times \mathbb{R}^{m})$ and $S_{\rho}^{\omega}(T^{*}M)$ consist of functions
satisfying (\ref{(38)}) in terms of local coordinates. If it is sufficiently clear what the underlying space is,
one writes simply $S_{\rho}^{\omega}$, and one sets
\begin{equation}
S_{\rho}^{\infty} \equiv \bigcup_{\omega \in \mathbb{R}} S_{\rho}^{\omega}, \qquad
S^{-\infty} \equiv \bigcap_{\omega \in \mathbb{R}} S_{\rho}^{\omega},
\label{(39)}
\end{equation}
where $S^{-\infty}$ is independent of $\rho$. 

Consider now a manifold $M$ endowed with a connection. We denote by $L_{\rho}^{\omega}(M)$ the space of operators on 
$C^{\infty}(M)$ which locally are pseudodifferential operators with symbols (as defined in (\ref{(18)})) in 
$S_{\rho}^{\omega}$. Given a linear function $l(v,x)$ as defined in section \ref{sec:l}, let 
$\psi \in C^{\infty}(M \times M)$ such that it is 1 on a neighbourhood of the diagonal and such that 
${\rm d}_{x}l(v,x) \neq 0$ for $\psi(x_{0},x) \neq 0$ and  $0 \neq v \in T_{x_{0}}^{*}$

Then, for any operator $A \in L_{\rho}^{\omega}$, define
\begin{equation}
\sigma_{A}(v)= \left . \Bigr[A \psi(\pi(v),x) {\rm e}^{{\rm i}l(v,x)}\Bigr]
\right|_{x=\pi(v)}.
\label{(41)}
\end{equation}
It can be proved \cite{ref:widom} that such a function belongs to $S_{\rho}^{\omega}$, and that different choices 
of functions $\psi$ and $l$ lead to the same function $\sigma_A$ modulo on element of $S^{-\infty}$. Therefore, 
the symbol $\sigma_A$ of the operator $A$ is defined as the corresponding equivalence class in $S_{\rho}^{\omega}/S^{-\infty}$. 

\section{Symbols in quantum field theory}

On the side of physical applications, let us here reconsider the photon propagator in the Euclidean
version of quantum electrodynamics. In modern language, the functional integral tells us that the photon 
propagator is obtained by first evaluating the gauge-field operator $P_{\mu \nu}$, $\mu,\nu = 0,\dots,3$, resulting from the
particular choice of gauge-averaging functional, then taking its symbol $\sigma(P_{\mu \nu})$ and inverting
such a symbol to find $\sigma^{-1}(P_{\mu \nu})=\Sigma^{\mu \nu}$ for which 
$\sigma \Sigma=\Sigma \sigma=I$. The photon propagator reads eventually \cite{ref:esp02}
\begin{equation}
\bigtriangleup^{\mu \nu}(x,y)=(2\pi)^{-4}\int_{\zeta}{\rm d}^{4}k \; \Sigma^{\mu \nu}
{\rm e}^{{\rm i}k \cdot (x-y)}
\label{(43)}
\end{equation}
for some contour $\zeta$. The gauge-field Lagrangian turns out to be
\begin{equation}
{\cal L}=\partial^{\mu}\rho_{\mu}+{1 \over 2}A^{\mu}P_{\mu \nu}A^{\nu},
\label{(44)}
\end{equation}
where 
\begin{equation}
\rho_{\mu}={1 \over 2}A_{\nu}\left(\partial_{\mu}A^{\nu}-\partial^{\nu}A_{\mu}\right)
+{1 \over 2 \alpha}A_{\mu}\partial^{\nu}A_{\nu},
\label{(45)}
\end{equation}
and
\begin{equation}
P_{\mu \nu}=-g_{\mu \nu} \Box + \left(1-{1 \over \alpha}\right)\partial_{\mu}\partial_{\nu}.
\label{(46)}
\end{equation}
Note that here we denoted by $\partial_{\mu}$ the standard partial derivative with respect to the $\mu$ component, 
that is $\partial_\mu = \frac{\partial}{\partial x^\mu}$. Also, $\alpha \in \mathbb{R}\setminus\{0\}$, $g_{\mu\nu} 
= \text{diag}(1,1,1,1)$ and $\Box = g^{\mu\nu} \partial_\mu \partial_\nu$. Of course, the term $\rho_{\mu}$ only 
contributes to a total divergence and hence does not
affect the photon propagator, while the parameter $\alpha$ can be set equal to $1$ (Feynman choice)
so that calculations are simplified. Thus, we can eventually obtain the gauge-field operator
\begin{equation}
P_{\mu \nu}(\alpha=1)=-g_{\mu \nu} \Box .
\label{(47)}
\end{equation}
Its symbol, which results from Fourier analysis of the $\Box$ operator, reads as
\begin{equation}
\sigma(P_{\mu \nu}(\alpha=1))=k^{2}g_{\mu \nu},
\label{(48)}
\end{equation}
and hence the Euclidean photon propagator is (cf. eq. (\ref{(4)}))
\begin{equation}
\bigtriangleup_{E}^{\mu \nu}(x,y)=(2\pi)^{-4}\int_{\Gamma}
{\rm d}^{4}k \; {g^{\mu \nu}\over k^{2}}
{\rm e}^{{\rm i}k \cdot (x-y)},
\label{(49)}
\end{equation}
where the points $x$ and $y$ refer to the indices $\mu$ and $\nu$, respectively. Note that integration
along the real axis for $k_{0},k_{1},k_{2},k_{3}$ avoids poles of the integrand, which are located at the
complex points for which $k^{2}=\sum_{\mu=0}^{3}(k_{\mu})^{2}=0$. 

Strictly speaking, the gauge parameter $\alpha$ in (\ref{(45)}) and (\ref{(46)}) is the bare value $\alpha_{B}$ of
$\alpha$ before renormalization, and we should express the bare symbol of the gauge-field operator in
QED in the form
\begin{equation}
\sigma(P_{\mu \nu})=k^{2}g_{\mu \nu}+\left({1 \over \alpha_{B}}-1 \right)k_{\mu}k_{\nu}
=\sigma_{\mu \nu}(k).
\label{(50)}
\end{equation}
Its inverse $\Sigma^{\mu \nu}$ is a combination of $g^{\mu \nu}$ and $k^{\mu}k^{\nu}$ with
coefficients $\cal A$ and $\cal B$, respectively, determined from the condition 
\begin{equation}
\sigma_{\mu \nu}\Sigma^{\nu \lambda}=\delta_{\mu}^{\; \lambda},
\label{(51)}
\end{equation}
which implies
\begin{equation}
{\cal A}={1 \over k^{2}}, \;
{\cal B}={(\alpha_{B}-1)\over k^{4}}.
\label{(52)}
\end{equation}
At this stage, the bare photon propagator takes the form
\begin{equation}
\bigtriangleup^{\mu \nu}(x,y)=\int_{\Gamma}{{\rm d}^{4}k \over (2\pi)^{4}}
\left[{g^{\mu \nu}\over k^{2}}+{(\alpha_{B}-1)k^{\mu}k^{\nu}\over k^{4}}
\right]{\rm e}^{{\rm i} k \cdot (x-y)}.
\label{(53)}
\end{equation}
 
\section{Hawking radiation}

Following Ref. \cite{ref:par} we now consider a completely different setting, i.e. the spacetime
of a collapsing star, for which, outside the horizon, the wave operator appears as it would in a
Schwarzschild spacetime for a suitable choice of coordinates. Let $r$ be such a radial coordinate with
associated hypersurface $\Sigma$ having equation $r=0$, and let $t$ be Killing time, so that 
$k={\partial \over \partial t}$ is the timelike Killing vector field outside the Killing horizon
$B_{k}$. The wave operator reads as
\begin{equation}
\Box ={\partial^{2}\over \partial t^{2}}-{1 \over r^{2}}\left(1-{2m \over r}\right)
{\partial \over \partial r}\left(r^{2}\left(1-{2m \over r}\right)
{\partial \over \partial r}\right)
-{1 \over r^{2}}\left(1-{2m \over r}\right)L,
\label{(54)}
\end{equation}
where the operator $L$ is independent of $t$ or $r$ and takes the same form as in Minkowski 
spacetime. The radial part of the wave operator is
\begin{equation}
\Box_{R}={\partial^{2}\over \partial t^{2}}
-{1 \over r^{2}}\left(1-{2m \over r}\right)
{\partial \over \partial r}\left(r^{2}\left(1-{2m \over r}\right)
{\partial \over \partial r}\right).
\label{(55)}
\end{equation}
By letting $\gamma_{d}$ (here $d$ is for detector, since a calculation along $\gamma_{d}$
yields the spectrum as measured by an observer whose worldline is $\gamma_{d}$) 
be an integral curve of ${\partial \over \partial t}$ and defining $\tau \equiv \kappa t$,
the radial part of the wave operator can be decomposed in the form
\begin{equation}
\Box_{R}=(\kappa {\dot \gamma}_{d})^{2}-\bigtriangleup_{R}
=(\kappa {\dot \gamma}_{d})^{2}
-{1 \over r^{2}}\left(1-{2m \over r}\right)
{\partial \over \partial r}\left(r^{2}\left(1-{2m \over r}\right)
{\partial \over \partial r}\right),
\label{(56)}
\end{equation}
where the dot denotes ${\partial \over \partial \tau}$. 

Since the operator $\sqrt{\bigtriangleup_{R}}$ is on firm ground, the operator $\Box_{R}$ can be decomposed as
\begin{equation*}
\Box_{R} = \left( \kappa {\dot \gamma}_{d} - \sqrt{\bigtriangleup_{R}} \right)  
\left( \kappa {\dot \gamma}_{d} + \sqrt{\bigtriangleup_{R}} \right).
\end{equation*}
In the following, we will investigate the spectral properties of the operator $\kappa {\dot \gamma}_{d} 
- \sqrt{\bigtriangleup_{R}}$, corresponding to the outgoing part of the radiation. Let us start by considering the 
first-order pseudodifferential operator (our conventions for numerical factors follow here
our Ref. \cite{ref:par})
\begin{equation}
P \equiv \sqrt{\bigtriangleup_{R}}+{\rm i}{\kappa \over 2\pi}\xi,
\label{(57)}
\end{equation}
obtained by taking the Fourier transform with respect to the variable $\tau$.

In order to study the spectrum of $\sqrt{\bigtriangleup_{R}}$, we consider the eigenvalue equation
$\bigtriangleup_{R}u+\lambda^{2}u=0$, where $u$ depends only on $r$. Upon denoting by a prime the
differentiation with respect to $r$, this reads as
\begin{equation}
\left(r^{2}\left(1-{2m \over r}\right)u' \right)'
+\lambda^{2}r^{2}\left(1-{2m \over r}\right)^{-1}u=0.
\label{(58)}
\end{equation}
This ordinary differential equation implies that $u(2m)=0$ while $u'(2m)$ is finite. At large $r$,
eq. (\ref{(58)}) reduces to 
\begin{equation}
(r^{2}u')'+\lambda^{2}r^{2}u=0.
\label{(59)}
\end{equation}
The general solution of eq. (\ref{(59)}) that is bounded at infinity is \cite{ref:par} 
\begin{equation}
u(r)=a {\sin(\lambda r) \over r}+b {\cos (\lambda r)\over r},
\label{(60)}
\end{equation}
where the parameters $a$ and $b$ are constant. From the condition $u(2m)=0$ one finds
\begin{equation}
-{b \over a}=\tan (2m \lambda),
\label{(61)}
\end{equation}
and \cite{ref:par}
\begin{equation} 
\lambda_{n}={\theta \over 2m}+n {\pi \over 2m}, \; n \in Z
\label{(62)}
\end{equation}
is the large-$r$ limit of the spectrum of $\bigtriangleup_{R}$, having set
$\theta \equiv \arctan (-b/a)$. Since the action of the wave operator on smooth functions should be
smooth on the hypersurface $\Sigma$, one also requires boundedness of $u$ as $r \rightarrow 0^{+}$.
This implies in turn that $b=\theta=0$ and $\lambda_{n}$ reduces to $n$ upon
rescaling the radial variable $r$. 

The spectral $\zeta$-function for the operator $P$ defined in (\ref{(57)}) is therefore
expressed, in the large-$r$ limit, by the asymptotic expansion
\begin{equation}
{\widetilde \zeta}_{P}(s) = \sum_{n=1}^{\infty}\left(n+{\rm i}{\kappa \over 2 \pi}\xi \right)^{-s},
\label{(63)}
\end{equation}
where the tilde is used to denote removal of the degeneracy of the vanishing eigenvalue. By relabelling 
the lower limit of summation, we can re-express this spectral $\zeta$-function in terms of the
Hurwitz $\zeta$-function, i.e. \cite{ref:par}
\begin{equation}
{\widetilde \zeta}_{P}(s) \sim \sum_{l=0}^{\infty}\left(l+1+{\rm i}{\kappa \over 2 \pi}\xi \right)^{-s}
=\zeta_{H}\left(s,1+{\rm i}{\kappa \over 2\pi}\xi \right).
\label{(64)}
\end{equation}

If $\omega$ is the variable dual to $\xi$ in the framework of Fourier transform, one can write
that the Fourier transform with respect to $\xi$ of ${\widetilde \zeta}_{P}(s)$ is approximated
by \cite{ref:dit}
\begin{equation}
{\cal F}({\widetilde \zeta}_{P}(s)) \sim 
{\left({2 \pi \over \kappa}\right)^{s}\omega^{s-1}\over \Gamma(s) 
({\rm e}^{2 \pi \omega / \kappa}-1)}.
\label{(65)}
\end{equation}
If one takes the limit as $s \rightarrow 1$ and recalls, from Ref. \cite{ref:ppe}, that
$\kappa={1 \over 4m}$ in the spacetime of a collapsing star, one finds 
\begin{equation}
{\cal F}({\widetilde \zeta}_{P}(s)) \sim 
{8 \pi m \over ({\rm e}^{8 \pi m \omega}-1)}.
\label{(66)}
\end{equation}
The spectral density $\rho$ is now given by
\begin{equation}
\rho = \kappa {\cal F}({\widetilde \zeta}_{P}(s))
\sim {2 \pi \over ({\rm e}^{8 \pi m \omega}-1)},
\label{(67)}
\end{equation}
which is the famous result of Hawking in Ref. \cite{ref:haw}.

\section{Concluding remarks}

Pseudodifferential \cite{ref:gru} and Fourier-Maslov   
integral operators \cite{ref:tre} play a key role in the modern theory of elliptic
and hyperbolic equations on manifolds, respectively, and the physical applications form an 
equally rich family, ranging from the Cauchy problem of classical field theory
\cite{ref:esp15} to the Green functions of quantum field theory and black-hole physics,
as we have shown.

Here we would like to add that, in the sixties, DeWitt discovered that the advanced and
retarded Green functions of the wave operator on metric perturbations in the de Donder gauge 
make it possible to define classical Poisson brackets on the space of functionals that are
invariant under the action of the full diffeomorphism group of spacetime. He therefore tried
to exploit this property to define invariant commutators for the quantized gravitational
field \cite{ref:dw}, but the operator counterpart of the classical Poisson brackets turned out
to be a hard task. On the other hand we know from section (\ref{sec:intro}) that, rather than inverting exactly 
a partial differential operator, it is more convenient to build a parametrix. This makes it possible
to solve inhomogeneous equations with the desired accuracy. Interestingly, it remains to be
seen whether such a construction might be exploited in canonical quantum gravity, provided one
understands what is the counterpart of classical smoothing operators in the quantization procedure.  

\section*{Acknowledgments}
G. Esposito is grateful to the Dipartimento di Fisica Ettore Pancini of Federico II University for
hospitality and support. G. M. Napolitano is grateful to the Department of Physics Ettore Pancini, University of
Naples Federico II, for the hospitality during the visit when this work was initiated.
The authors dedicate the present review paper to Gaetano Vilasi, who has done an 
invaluable work for several generations of students and researchers at Salerno University.

\begin{appendix}

\section{Symbol map for partial differential operators; the space of pseudodifferential
operators on $\mathbb{R}^n$}

The definitions of Section $1$ are best suited to deal with the analysis of inhomogeneous partial
differential equations $Pu=f$ and use a nomenclature very close to the one appropriate for
Fourier-Maslov integral operators for hyperbolic equations. However, we should also recall the basic
properties summarized below \cite{ref:gil}.

A linear partial differential operator $P$ of order $d$ on $\mathbb{R}^n$ is a polynomial expression
\begin{equation}
P(x,D)=\sum_{\alpha: |\alpha| \leq d}a_{\alpha}(x)D_{x}^{\alpha}, \; 
a_{\alpha} \in C^{\infty}(\mathbb{R}^n),
\label{(A1)}
\end{equation}
where, for the multi-index $\alpha=(\alpha_{1},\dots,\alpha_{n})$, the modulus $|\alpha|$
and the derivative operator $D_{x}^{\alpha}$ are defined by
\begin{equation}
|\alpha| \equiv \alpha_{1}+\cdots+\alpha_{n}, \;
D_{x}^{\alpha} \equiv (-{\rm i})^{|\alpha|}
\left({\frac{\partial}{\partial x_{1}}}\right)^{\alpha_{1}} \cdots
\left({\frac{\partial}{\partial x_{n}}}\right)^{\alpha_{n}}.
\label{(A2)}
\end{equation}
The symbol of $P$ is then defined by
\begin{equation}
\sigma(P)=\sigma(x,\xi) \equiv \sum_{\alpha: |\alpha| \leq d}a_{\alpha}(x)\xi^{\alpha},
\label{(A3)}
\end{equation}
and is a polynomial of order $d$ in the dual variable $\xi$, where $(x,\xi)$ defines a point of
the cotangent bundle of $\mathbb{R}^n$. The {\it leading symbol} 
is the highest order part of $\sigma(x,\xi)$, i.e.
\begin{equation}
\sigma_{L}(P)=\sigma_{d}(x,\xi) \equiv \sum_{\alpha: |\alpha|=d}a_{\alpha}(x)\xi^{\alpha},
\label{(A4)}
\end{equation}
and the action of $P$ can be re-expressed in integral form as
\begin{equation}
Pf(x)=(2\pi)^{-{\frac{n}{2}}} \int_{\mathbb{R}^n}
{\rm e}^{{\rm i}\xi \cdot x} \sigma(x,\xi){\hat f}(\xi){\rm d}\xi.
\label{(A5)}
\end{equation}

In general, one can consider the set $S^{d}$ of all symbols $\sigma(x,\xi)$ such that
\vskip 0.3cm
\noindent
(i) $\sigma$ is smooth in $(x,\xi)$ with compact $x$ support.
\vskip 0.3cm
\noindent
(ii) For all multi-indices $(\alpha,\beta)$, there exist constants $C_{\alpha,\beta}$ for which
\begin{equation}
\left | D_{x}^{\alpha}D_{\xi}^{\beta} \sigma(x,\xi) \right| \leq C_{\alpha,\beta}
(1+|\xi|)^{d-|\beta|}.
\label{(A6)}
\end{equation}

For $\sigma \in S^{d}$, one defines the associated operator (${\cal S}$ being the Schwartz space
of smooth complex-valued functions with fast decrease) 
$P: {\cal S} \rightarrow C_{0}^{\infty}(\mathbb{R}^n)$ as in (A5), i.e.
\begin{equation}
Pf(x) \equiv (2\pi)^{-{\frac{n}{2}}} \int_{\mathbb{R}^n}
{\rm e}^{{\rm i}\xi \cdot x}\sigma(x,\xi){\hat f}(\xi){\rm d}\xi
=(2\pi)^{-n}\int_{\mathbb{R}^n}
{\rm e}^{{\rm i}\xi \cdot (x-y)}\sigma(x,\xi)f(y){\rm d}y \; {\rm d}\xi.
\label{(A7)}
\end{equation}
The space $\psi^{d}$ of such operators is the set of pseudo-differential operators of order $d$
\cite{ref:see}.    

\end{appendix}

\bibliographystyle{unsrt}

\end{document}